\documentclass[letterpaper, 10pt, conference]{ieeeconf}
\IEEEoverridecommandlockouts
\usepackage{amsmath,amsfonts,amssymb}
\usepackage{graphicx}
\usepackage{xcolor}
\usepackage{mathtools}
\usepackage{nicefrac}
\usepackage{tikz, pgfplots, pgfplotstable}
\usepackage{glossaries}
\usepackage{xspace}
\usepackage{algpseudocode}
\usepackage{algorithm}
\usepackage{booktabs}
\usepackage{cite}
\usetikzlibrary{positioning}
\usetikzlibrary{decorations.pathreplacing}
\usepgfplotslibrary{fillbetween,groupplots}
\usetikzlibrary{patterns,patterns.meta}

\let\labelindent\relax 
\usepackage{enumitem}

\usepackage[table]{xcolor}
\definecolor{lightgray}{gray}{0.9}
\let\oldtabular\tabular
\let\endoldtabular\endtabular
\renewenvironment{tabular}{\rowcolors{2}{white}{lightgray}\oldtabular}{\endoldtabular}
\usepackage{multirow}

\pgfplotsset{compat=1.18}

\newtheorem{assumption}{Assumption}
\newtheorem{definition}{Definition}
\newtheorem{theorem}{Theorem}
\newtheorem{lemma}{Lemma}
\newtheorem{remark}{Remark}

\newtheorem{example}{Example}
\newcommand{\mm}{\theta}

\newcommand{\nstates}{n}
\newcommand{\ninputs}{m}
\newcommand{\nstatest}{q}

\newcommand{\xt}{{x}}
\newcommand{\x}{\tilde{x}}
\newcommand{\auxvar}{z}
\newcommand{\SOL}{P}

\newcommand{\nconstu}{w}
\newcommand{\nconstx}{v}

\newcommand{\el}{\ell}

\newcommand{\Xf}{\mathcal{X}^\mathrm{f}}
\newcommand{\Xfi}{\mathcal{X}^\mathrm{f}_i}

\newcommand{\Vf}{J^\mathrm{f}}
\newcommand{\VN}{J_N}
\newcommand{\VNstar}{\VN^\star}
\newcommand{\VNhat}{\hat{J}_N}

\newcommand{\numNodes}{{n_\mathcal{V}}}
\newcommand{\numEdges}{{n_\mathcal{E}}}

\newcommand{\dec}{d}
\newcommand{\decset}{\mathcal{D}}

\newcommand{\setKll}{\Theta_\textsc{kll}}
\newcommand{\setCsn}{\Theta_\textsc{csn}}

\newcommand{\setN}{\mathcal{N}^{\theta_0,\dec}}

\newcommand{\stree}[1]{\mathcal{T}(#1)}

\newacronym{mjls}{MJLS}{Markov jump linear system}
\newacronym{ocp}{OCP}{optimal control problem}
\newacronym{mpc}{MPC}{model predictive control}
\newacronym{mss}{MSS}{mean-square stable}
\newacronym{lmi}{LMIs}{linear matrix inequalities}
\newacronym{let}{LET}{logical execution time}

\definecolor{firstcol}{HTML}{EE7733}
\definecolor{secondcol}{HTML}{0077BB}
\definecolor{thirdcol}{HTML}{99ddf6}
\definecolor{offwhite}{HTML}{FAEBD7}
\definecolor{greenlightcol}{HTML}{CCDDAA}

\newcommand{\colvec}[2][.8]{{
\setlength\arraycolsep{2pt}\scriptstyle\begin{bmatrix}#2\end{bmatrix}
}}

\pgfplotsset{ every non boxed x axis/.append style={x axis line style=-},
     every non boxed y axis/.append style={y axis line style=-}}

\newcommand\copyrighttext{%
  \footnotesize \textcopyright 2026 IEEE. Personal use of this material is permitted. Permission from IEEE must be obtained for all other uses, including reprinting/republishing this material for advertising or promotional purposes, collecting new collected works for resale or redistribution to servers or lists, or reuse of any copyrighted component of this work in other works.}
\newcommand\copyrightnotice{%
\begin{tikzpicture}[remember picture,overlay]
\node[anchor=south,yshift=5pt] at (current page.south) 
  {\fbox{\parbox{\dimexpr\textwidth-\fboxsep-\fboxrule\relax}{\copyrighttext}}};
\end{tikzpicture}%
}

\begin{document}

\title{Adaptive Handling of Deadline Misses via Stochastic MPC}

\author{
Melanie Gallant$^{1,2}$,
Paolo Pazzaglia$^{1}$,
Christoph Mark$^{1}$,
Kevin Schmidt$^{1}$,
Martina Maggio$^{2,3}$
\thanks{$^{1}\,$Robert Bosch GmbH, Germany, \texttt{\{name.surname\}@de.bosch.com}}
\thanks{$^{2}\,$Saarland University, Germany, \texttt{maggio@cs.uni-saarland.de}}
\thanks{$^{3}\,$Lund University, Sweden.
\newline \textbf{Acknowledgements:} This work was supported by the German Federal Ministry for Economic Affairs \& Climate Action under grant no.~13IPC021.}
}

\maketitle
\copyrightnotice
\begin{abstract}
Real-time control systems may experience sporadic computational overloads that cause controller executions to overrun their expected deadlines, affecting the timely update of control inputs. Such overruns are classically handled with static decisions, e.g., by either stopping and discarding the late execution at the deadline instant, or continuing it while skipping the next activations. This paper first demonstrates that such static policies can lead to suboptimal performance when the controlled system undergoes varying operating conditions, and thus proposes a control-driven hybrid overrun strategy that optimally adapts the  deadline miss handling at runtime. We model the closed loop as a Markov jump linear system, whose modes encode overrun histories, and we formulate a scenario-tree-based stochastic model predictive control problem that predicts and optimizes over possible sequences of deadline hits and misses. We evaluate our proposal on linear benchmarks including a double integrator and an aircraft pitch system.
\end{abstract}

\section{Introduction}

Commercial control systems must balance the conflicting needs of both adequate performance and price competitiveness.
These considerations impact the choice of the real-time-capable computational platform  used for implementing the controller.
As a result, it is not uncommon for the chosen platforms to handle only nominal operating conditions without problems:
This means that, in case of computational overloads, the periodic control signal calculation might not terminate always within the expected real-time deadline --an event denoted as an \emph{overrun}, or \emph{deadline miss} (with \emph{deadline hit} instead denoting a timely computation).
In industry, overruns are indeed fairly common~\cite{akesson2020empirical}, even in well-designed systems, due to complex hardware- and software-related dependencies. 
This is particularly relevant when dealing with computationally intensive control algorithms such as \gls{mpc}~\cite{rawlings2017model}, an increasingly popular solution to explicitly address control constraints such as actuator limitations and safety boundaries.

During overruns, the physical system might not receive timely control inputs, potentially affecting its performance and stability~\cite{pazzaglia2018beyond, vreman2021stability}.
Nonetheless, a properly designed system should be able to enforce a safe behavior~\cite{hobbs2022safety} and properly react to overruns. 
State-of-the-art overrun handling strategies generally require either ``\emph{killing}'' the computation that missed its deadline (and rolling back any partial computation) to start a new iteration, or letting the late execution ``\emph{continue}'' until completion avoiding the start of a new one, which is in turn either delayed or entirely skipped~\cite{CERVIN2005219, pazzaglia2019dmac}.

The {overrun strategy} has been generally treated as a system property~\cite{gallant2026comparing}, i.e., considered fixed, and usually selected by analyzing the closed-loop system properties to assess stability and performance guarantees.
In this paper, we argue that this approach has historically overlooked the influence of different operating conditions of the physical plant, by testing only specific step-response scenarios.
Indeed, our findings show that there are categories of dynamical systems where using one fixed overrun handling strategy in different operating conditions (such as different kinds of physical disturbances or different overload patterns) leads to suboptimal performance.
On the contrary, a more flexible approach that can switch between different strategies at runtime allows for better closed-loop performance.

\paragraph*{Contribution}
We propose a runtime-adaptive stochastic scenario-tree \gls{mpc} approach for real-time control applications which: i) decides at runtime which overrun strategies to apply in the next steps that might be affected by overruns; and ii)~computes the optimal control input considering the system evolution under different overrun strategies and potential deadline misses over a finite window.
To the best of our knowledge, this is the \textit{first proposal} of a predictive non-static overrun strategy for a real-time control system.

To this end, we model the controlled system as a \gls{mjls}, with a probabilistic characterization of deadline misses.
We make use of a {scenario-tree} \gls{mpc}, by casting the possible state evolutions over a finite horizon of the future deadline hits and misses under our hybrid strategies.
The \gls{mpc} iteratively computes the optimal control sequence and predicts the state evolution, which is used in an optimization algorithm to select the overrun strategies for the next steps that might be affected by deadline misses.
The design is then certified correct for the standard requirements of recursive feasibility and stability.
Finally, we test our method on multiple case studies, including an aircraft pitch system, showing improvements when using our hybrid overrun strategy against fixed ones.
We also propose pruning strategies to significantly reduce the computational load with minimal impact on control performance.

\paragraph*{Related Work}
A vast literature explores the probabilistic characterization of deadline misses in real-time systems~\cite{davis2019survey}.  Particularly relevant to this paper, Chen et al.~\cite{chen2024dawn} analyze soft real-time tasks where jobs can execute after their deadline until a dismissal point.
Deadline miss probabilities with Markov chain models are explored e.g., in~\cite{friebe2024efficiently} and in~\cite{frias2017probabilistic}. 
Weakly-hard models~\cite{bernat2001weaklyhard} are also commonly used to analytically bound the maximum number of deadline misses in a given time window. 
Bounding probabilities for multiple deadline misses in a row was explored in~\cite{chen2017probabilistic} using analytical bounds, and in~\cite{pazzaglia2019dmac} using scenario theory. 

As mentioned above, such systems require overrun handling strategies~\cite{CERVIN2005219, pazzaglia2019dmac}. Most works in literature employ a fixed strategy, with just few examples of mixed ones: 
A double-threshold approach is used in~\cite{fontanelli2013soft}, where missing the first deadline let the execution continue until a second threshold, which triggers its killing.
A mixed nominal-skip-abort strategy for networked control systems is presented in~\cite{soudbakhsh2016co}, but with delays shorter than one period.
A recent usage of mixed strategy was presented as an example in~\cite{gallant24dr} in a distributionally robust control design for \gls{mjls}s with mixed known, unknown and time varying probabilities of overruns.
In all such cases, no runtime adaptation is used:
our work is the first one to propose an online overrun strategy adaptation, with the additional benefit of being based on a predictive optimal control scheme.

Alternative runtime solutions to overloads focus on different perspectives. 
Choi et al.~\cite{choi2019job} proposed a dynamic job-class level priority assignments to enforce weakly-hard constraints.
Proactively skipping jobs to enforce weakly-hard constraints has also been explored, e.g., in~\cite{yoshimoto2011optimal, wang2021cross}. 
An example of dynamic scheduling strategy for overloaded systems is presented in~\cite{gao2022dynamic}.
Alternatively, also adapting at runtime the controller in case of overruns has been explored e.g., in~\cite{pazzaglia2021adaptive,vreman2022deadline}, but using simple feedback control designs. 

Recently, many works tackled the problem of safety guarantees for overloaded real-time control systems, such as \cite{huang2019formal, wang2021cross, hobbs2022safety, xu2023safety, gifford2024decntr}. All such works guarantee state evolution within a safety boundary in case of deadline misses.
In this paper, by making use of scenario-tree-based \gls{mpc}, safety in the presence of overruns can be intrinsically tackled, by directly encoding safety constraints in its formulation.

\section{Basic model and background}

We consider a discrete-time linear system
\begin{equation}\label{eq:phys}
    \x(k+1) = A\, \x(k) + B\, u(k), \quad \x(0) = \x_0,
\end{equation}
where $\x \in \mathbb{R}^{\nstates}$ is the state, $u \in \mathbb{R}^{\ninputs}$ is the input, and the matrices $A \in \mathbb{R}^{\nstates \times \nstates}$ and $B \in \mathbb{R}^{\nstates \times \ninputs}$ characterize the system evolution at discrete time steps $k\in\mathbb{N}$. 
The inputs and states are subject to polytopic constraints 
\begin{subequations}
\label{eq:constraints}
\begin{align}
   u(k) &\in \mathcal{U} = \left\{ u\in\mathbb{R}^\ninputs \;|\; H\,u \leq h \right\} \quad &\forall k \in \mathbb{N} \label{eq:constraints:input} \\
   \x(k) &\in \mathcal{X} = \left\{ \x\in\mathbb{R}^{\nstates} \;|\; G\,\x \leq g \right\} \quad &\forall k \in \mathbb{N} \label{eq:constraints:state}
\end{align}
\end{subequations}
with $H\in \mathbb{R}^{\nconstu\times\ninputs}$, $h \in \mathbb{R}^{\nconstu}$, $G\in\mathbb{R}^{\nconstx\times\nstates}$ and $g\in\mathbb{R}^{\nconstx}$.

\paragraph*{Real-time controller}

We employ \gls{mpc} to optimize in a predictive fashion the system input $u(k)$ that minimizes a quadratic cost term, and
regulates the state of system~\eqref{eq:phys} to its origin while satisfying the constraints \eqref{eq:constraints}. The controller operates as a real-time periodic task with a sampling period synchronized with that of system~\eqref{eq:phys}.
At the start of every controller period $k \in \mathbb{N}$, an instance of the controller task (\emph{job}) is triggered.
The $k$-th job executing during $[k, k+1)$ uses the state $\x(k)$ to compute a new control input $u(k+1)$ to be applied at the job's deadline.
Following the implicit deadline model, the $k$-th controller job has a deadline that coincides with the activation of the $(k+1)$-th controller job.

\paragraph*{Overruns and classic handling strategies}
If the controller task execution does not finish within its deadline, and thus experiences an \emph{overrun}, a timely valid value for $u(k+1)$ is not available.
Following standard assumptions, $u(k+1)$ is either set to zero (\textsc{zero} actuation strategy), or kept constant (\textsc{hold}) with the previously available value~\cite{schenato2009zero}.

In a real-time setting, overrunning jobs are treated as exceptions that need to be properly handled.
In literature, three overrun strategies are commonly discussed: \emph{kill, continue-and-skip-next}, and \emph{queue}~\cite{CERVIN2005219}.
Under kill, the job is terminated at the missed deadline, its partial computations discarded, and a new job is activated.
Under continue-and-skip-next, the overrunning job is continued until completion, while discarding the jobs that would have been activated in the next periods (until the late job completes).
Under queue, the overrunning job is continued until completion, while new jobs are added to a finite buffer (e.g., of size $1$) and start executing as soon as the late job finishes.
For analysis purposes, queue has been considered far less due to cascade delays and non-deterministic realizations, hence we focus our attention to kill (denoted hereafter as \textsc{kll}) and continue-and-skip-next (\textsc{csn}).

We consider a generic stochastic model for overruns, where the probabilities of deadline misses under \textsc{kll} and \textsc{csn} are available, e.g., by estimation based on past observations.
If only limited data is available, a distributionally robust approach as in \cite{gallant2026approximate} could be employed to explicitly take into account the estimation uncertainty.
We also assume that the number of consecutive periods without a completed job can be safely bounded by a positive integer $n_\mathrm{M}$.\footnote{The value $n_\mathrm{M}$ can be computed, e.g., by using a \emph{weakly-hard} analysis~\cite{bernat2001weaklyhard}. In general the value of $n_\mathrm{M}$ under \textsc{csn} is lower or equal the value under \textsc{kll} -- see for example the results presented in~\cite{maggioecrts,vreman2021stability}. We will here consider the largest value.}

\section{Adaptive overrun strategy}

This section describes our novel control and overrun handling approach that adaptively decides between \textsc{kll} and \textsc{csn} for potential overruns during the next few periods.

\subsection{System model with hybrid overrun strategy}

Our method iteratively enforces a hybrid sequence of overrun strategies at runtime, applied in case of overruns occurring before the next successfully completed control job.

\begin{definition}
\label{def:predictive_decision_sequence}
Let $n_\mathrm{M} \in \mathbb{N}$ denote the bound on consecutive overruns between time $k$ and the next job completion. A \emph{predictive decision sequence} of length $n_\mathrm{M}$ is defined as
\begin{align*}
    \dec(k) = \langle \dec_{k}, \dec_{k+1}, \dots, \dec_{k+n_\mathrm{M}-1} \rangle \in \{\textsc{kll}, \textsc{csn}\}^{n_\mathrm{M}},
\end{align*}
where each $\dec_{k+i}$ for $i \in \{0, \ldots, {n_\mathrm{M}-1} \}$ represents a decision regarding an overrun strategy, chosen from either $\textsc{kll}$ or $\textsc{csn}$. 
The set of all possible predictive decision sequences is denoted by $\decset$, with cardinality $|\decset| = 2^{n_\mathrm{M}}$.
\end{definition}

The combination of the plant~\eqref{eq:phys}, its controller, and the effect of deadline misses under our hybrid overrun strategy, can be modeled as a discrete-time \gls{mjls}
\begin{equation}
\label{eq:mjls}
\begin{aligned}
    \xt(k+1) &= A_{\mm(k)}\, \xt(k) + B_{\mm(k)}\, u(k), \\
    \xt(0) &= \xt_0, \; \mm(0)=\mm_0,
\end{aligned}
\end{equation}
where $\xt(k) \in\mathbb{R}^{\nstatest}$ is an extended state and  $u(k) \in \mathbb{R}^\ninputs$ the control signal applied during the $k$-th period. 
The system switches between a finite set $\Theta$ of pairs of properly sized matrices $\{ A_i, B_i \}_{i \in \Theta}$, which model the behavior of~\eqref{eq:phys} in our hybrid overrun strategy setup. The evolution of the switching signal $\mm(k)\in\Theta$ is governed by a discrete-time Markov chain defined on the probability space $(\Theta,2^\Theta,\mathbb{P})$.

To fully characterize the switching behavior of the system, we introduce the following Markov modes corresponding to \emph{(i)}~a time step that ended with a job completion ($\textsc{h}$); and sequences of consecutive overruns having at the current time step either \emph{(ii)}~a missed deadline with decision of \textsc{csn} ($\textsc{c}$) or \emph{(iii)}~a missed deadline with decision of \textsc{kll} ($\textsc{k}$).
We can then formally define the set $\Theta$ as
\begin{equation}\label{eq:markovset}
\Theta \coloneqq \{\textsc{h}\} \cup \bigcup_{k=1}^{n_\mathrm{M}} \{\textsc{k},\textsc{c}\}^k,
\end{equation}
where $|\Theta|=2^{n_\mathrm{M}+1}-1$.
For compactness, we introduce the notation $\star$\textsc{k} and $\star$\textsc{c} to represent all Markov modes in $\Theta$ ending in \textsc{k} and \textsc{c}, respectively (e.g., for $n_\mathrm{M}=2$, $\star\textsc{k}:= \{\textsc{k}, \textsc{kk},\textsc{ck}\}$).
Each decision immediately constrains the reachable modes, i.e., if $d_k=\textsc{kll}$ then, at the next time step, $\theta(k+1) \in \{\textsc{h},\star\textsc{k}\}$.
We define the reduced sets of modes $\Theta_d$ as $\setKll=\{\textsc{h},\star\textsc{k}\}$ and $\setCsn=\{\textsc{h},\star\textsc{c}\}$.
The matrix $\Pi(d_k) = \left[\pi_{ij}(d_k)\right]_{i,j\in{\Theta}}$ encodes the transition probabilities $\pi_{ij}(d_k)=\mathbb{P}\left[\mm(k+1)=j\,|\,\mm(k)=i, \, d_k\right]$ between the Markov modes, affected by the decision with 
\begin{align} \label{eq:transition_prob_d}
\pi_{ij}(d_k) = \begin{cases}
    p_{ij}, & j\in \Theta_{d_k} \\
    0, & j\notin \Theta_{d_k},
\end{cases}
\end{align}
and for the not yet decided steps with
\begin{align}
    \pi_{ij}(\varnothing) = \bar{p}_{ij}.
    \label{eq:transition_prob_no_dec}
\end{align}
The transition probability to the $\textsc{h}$ mode only depends on the predecessor mode and is independent of the decision, i.e., $\pi_{i,\textsc{h}}(d_k) = \pi_{i,\textsc{h}}(\varnothing)=\pi_{i,\textsc{h}}$, for all $i\in\Theta$.
The decision remains open for future predicted steps after a hit, when a new decision sequence will be selected.
In that case, the distribution of transition probabilities to $\star\textsc{k}$ and $\star\textsc{c}$ is weighted by constant real-valued weighting factor $a\in(0,1)$, e.g., $\bar{p}_{\textsc{h},\textsc{k}} = a(1-p_{\textsc{h},\textsc{h}})$ and $\bar{p}_{\textsc{h},\textsc{c}} = (1-a)(1-p_{\textsc{h},\textsc{h}})$.

Figure~\ref{fig:graph} shows an example of the Markov chain for our hybrid strategy model, with $n_\mathrm{M}=2$ and thus $7$ modes.
Overrun modes (i.e., not $\textsc{h}$ modes) encode their consecutive overrun history (prefix) as well as the current overrun.
They are only reachable from the overrun mode defined by their prefix, or from the $\textsc{h}$ mode, if there is no prefix.
Given a decision sequence $d$ and after pruning the graph of modes not reachable from \textsc{h}, the graph is stochastic, i.e., the sum of the weights of all outgoing edges from a mode is one.
E.g., for $d=\langle \textsc{kll}, \textsc{csn}\rangle$, only the modes $\Theta = \{ \textsc{h}, \textsc{k}, \textsc{kc}\}$ would remain, and
$\pi_{\textsc{h},\textsc{k}}(d_k) = 1-\pi_{\textsc{h},\textsc{h}}$ and $\pi_{\textsc{k},\textsc{kc}}(d_{k+1}) = 1-\pi_{\textsc{k},\textsc{h}}$.

\begin{figure}[t]
    \centering
    \includegraphics{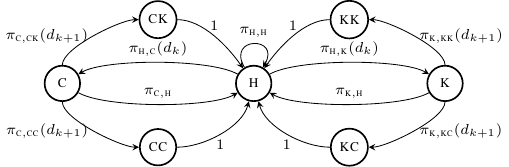}
    \vspace{-5mm}
    \caption{Graph of switching signal with $n_\mathrm{M}=2$.}
    \label{fig:graph}
\end{figure}
\begin{figure}[t]
    \centering 
    \includegraphics{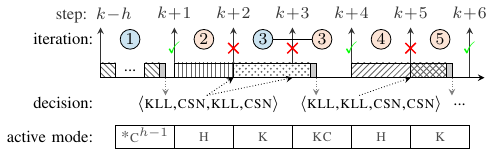}
    \caption{Execution example with $n_{\mathrm{M}}=4$. Job numbers are orange if executed under \textsc{kll}, and blue under \textsc{csn}. A job executing over multiple periods can be subject to different strategies (e.g., job \raisebox{.5pt}{\textcircled{\raisebox{-.8pt} {3}}} is first under \textsc{csn} and then under \textsc{kll}).}
    \label{fig:task}
\end{figure}

Finally, we formally introduce the matrices of the dynamics associated to the Markov modes defined in~\eqref{eq:markovset}.
For this purpose, we define the extended state as $\xt = [\x^\top \,\, \auxvar_1^\top \,\, \auxvar_2^\top]^\top$, where the auxiliary variable $\auxvar_1(k) \in \mathbb{R}^\ninputs$ stores the most recent value of $u$ from time step $k-1$ and $\auxvar_2(k) \in \mathbb{R}^\ninputs$ contains the output of the control job that was active during the previous period.\footnote{This compact formulation, inspired from~\cite{gallant2025softMPC}, allows to use the same matrix pair $(A_{\textsc{H}},B_{\textsc{H}})$ for both the case of a time step where a job completes within its first deadline, and for the case of a late job completing (often called in literature \emph{recovery hit}). For this reason, the Markov mode \textsc{H} encodes here both a timely and late job completion.}
To model the impact of different overruns onto the dynamics~\eqref{eq:mjls}, we define for each Markov mode an individual pair of system matrices $(A_i, B_i)$:
\begin{alignat*}{6}
    A_{\textsc{h}} = \colvec{
        A & 0 & B \\
        0 & 0 & I \\
        0 & 0 & 0 \\
    },\phantom{^\top}\; &
    A_{\star\textsc{k}} = \colvec{
        A & B & 0 \\
        0 & \Delta & 0 \\
        0 & 0 & 0
    },\phantom{^\top}\; &
    A_{\star\textsc{c}} = \colvec{
        A & B & 0 \\
        0 & \Delta & 0 \\
        0 & 0 & I \\
    },\phantom{^\top} \\
    B_{\textsc{h}} = \;\colvec{
        0 & 0 & I
    }^\top,\; &
    B_{\star\textsc{k}} = \;\colvec{
        0 & 0 & I
    }^\top,\; &
    B_{\star\textsc{c}} = \;\colvec{
        0 & 0 & 0
    }^\top,
\end{alignat*}
where $\Delta\in\mathbb{R}^{m\times m}$ is either the identity matrix ($I$ for \textsc{hold}) or a matrix of zeros ($0$, for  \textsc{zero}).

\begin{example}
    Figure~\ref{fig:task} shows an example of execution of the control task under our approach. The first shown iteration started at step $(k-h)$ and completes during step $k$. Then the hybrid overrun strategy decision process (depicted in gray) is executed and outputs the predictive decision sequence $\dec(k+1)=\langle\textsc{kll},\textsc{csn},\textsc{kll},\textsc{csn}\rangle$, to be used from step $k+1$ for up to $n_{\mathrm{M}}=4$ consecutive overruns. During the $(k+1)$-th step the job does not complete and it is therefore killed, following $\dec_0 = \textsc{kll}$. A new job is then activated at step $(k+2)$, again not completing in time. Since $\dec_1 =\textsc{csn}$, that job is allowed to continue and no new activation takes place at step $(k+3)$. The active job completes before $(k+4)$, and a new decision is made:  $d(k+4)=\langle\textsc{kll},\textsc{kll},\textsc{csn},\textsc{csn}\rangle$. Then, the job started at $(k+4)$ does not complete and is killed, while the following one completes its execution, and a new decision is provided. 
    The last line in Figure \ref{fig:task} represents the active Markov mode, identifying what happened in the relevant previous periods (either a successful completion, or overruns dealt with different strategies).
\end{example}

\subsection{Scenario-tree representation}

\begin{figure}[t]
    \centering
    \includegraphics{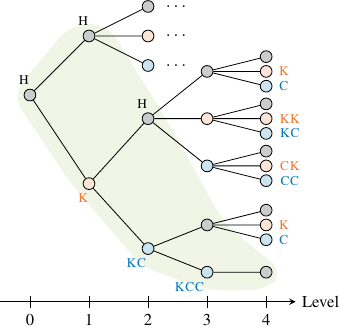}
    \caption{Scenario tree with $N=4$, $n_\mathrm{M}=3$, $\theta_0=\textsc{h}$ and $\dec=\langle\textsc{kll},$ $\textsc{csn},\textsc{csn}\rangle$. Nodes of mode $\star\textsc{k}$/$\star\textsc{c}$ are colored orange/blue. The green area highlights the nodes where the chosen $\dec$ holds.}
    \vspace{-2mm}
    \label{fig:scenario_trees}
\end{figure}

Similarly to \cite{gallant2025softMPC}, we cast the predicted behavior of \eqref{eq:mjls} into \emph{scenario trees}. 
\begin{definition}
A scenario tree $\mathcal{T}$ is an acyclic directed weighted graph $(\mathcal{V}, \mathcal{E})$ with nodes (vertices) $\mathcal{V} = \{v_0,\dots,v_{\numNodes-1}\}$ and edges $\mathcal{E} = \{e_0,\dots,e_{\numEdges-1}\}$. 
Each node is a $4-$tuple $v_\ell = (\xt_\ell,u_\ell,\theta_\ell,p_\ell) \in \mathcal{V}$ that contains the corresponding state, input, Markov mode, and absolute probability of occurrence.
We define $\mathcal{N} = \{0,\dots,\numNodes-1\}$ as the set of the tree's node indices and $\mathcal{N}_l$ the index set of nodes at a certain \textit{level} $l$ of the tree.
An edge between node $v_\ell$ to node $v_\varrho$ is weighted with the transition probability between the associated Markov modes.
The successor node $v_{\varrho}$ of node $v_\ell$ with mode $i$, given the edge exists, is identified with the function $\varrho=\mathtt{succ}(\ell,i)$.
A node's probability is computed by composing the probabilities of the edges along its unique path from the root node, starting from $p_0=1$.
\end{definition}

Given a prediction horizon $N\geq n_\mathrm{M} + 1$, we construct separate tree structures $\stree{\theta_0,\dec}=\langle \mathcal{V}^{\theta_0,\dec}, \mathcal{E}^{\theta_0,\dec}\rangle$ with node indices $\setN$, for all root nodes $v_0$ having modes $\theta_0\in\{\textsc{h},\star\textsc{k}\}$ (no new control job is activated when in modes $\star\textsc{c}$) and for all decision sequences $\dec\in\decset$.
Figure \ref{fig:scenario_trees} shows an example of such a tree.
The trees differ due to the pruning (highlighted area in the Figure) of the Markov chain produced by applying \eqref{eq:transition_prob_d} for a specific decision sequence $d$. 
When a node with mode $\textsc{h}$ is reached, a new (future) decision $\dec^+$ is expected, thus after such a node the tree spans all possible decisions (region outside the highlighted area), with edge weights determined by \eqref{eq:transition_prob_no_dec}.

\begin{definition}
A set of trees comprising all trees with root node of mode $\theta$ is called a forest $\mathcal{F}(\theta) = \{ \stree{\theta,\dec} \mid \dec\in\decset \}$ and the forest of all constructed tree structures is defined as $\mathcal{F} =  \bigcup_{\theta \in \{\textsc{h},\star\textsc{k}\}} \mathcal{F}(\theta) $.
\end{definition}
Overall, the number of generated scenario trees is $|\mathcal{F}| = 2^{n_\mathrm{M}} + n_\mathrm{M}\cdot 2^{n_\mathrm{M}-1}$ (obtained by adding the trees $\forall\dec\in\decset$ with $\theta_0=\textsc{h}$ and $\theta_0\in\{\star \textsc{k}\}$), for each of which a separate \gls{ocp} is compiled offline. The complexity of such offline setup is $\mathcal{O}(2^{n_\mathrm{M}})$.
Further, we use a mapping function to portray the time-shift of scenario trees given a specific sequence of modes in between started control jobs:

\begin{definition}
\label{def:node_index_mapping}
Let $h \in \mathbb{N}$ be a positive integer representing a time shift. We define the function $\Gamma: \mathcal{N}^{\theta(k),\dec} \to \mathcal{N}^{\theta(k-h),\dec}$ as a mapping of node indices from a tree at time $k$ to a tree at time $k-h$. Specifically, $\Gamma$ maps node indices from the tree $\stree{\theta(k),\dec}$ to the corresponding node indices in the tree $\stree{\theta(k-h),\dec}$, where $\stree{\theta(k-h),\dec}$ refers to the tree used in the last started control job.
\end{definition}

\begin{remark}
\label{rem:g_dependency}
The function $\Gamma$ inherently depends on the specific path traversed within the decision tree over the time interval $[k-h, k]$. This path is uniquely determined by the sequence of past decisions $d$ and the historical mode sequence $[\theta(k-h), \dots, \theta(k)]$. 
This is crucial for establishing the correct correspondence between node indices in the current tree $\stree{\theta(k),\dec}$ and the nodes of the previous tree $\stree{\theta(k-h),\dec}$, from which the last control job was initiated.
\end{remark}

\subsection{Hybrid overrun scenario-tree MPC}
\label{sec:hybridscenariotreempc}
The control task uses the scenario tree $\stree{\theta(k),\dec(k)}$ introduced above to propagate the system trajectories from the current state $\xt(k)$ over all possible branches, and to compute the optimal control input for the expected mode sequence with horizon $N$.
The control objective of the \gls{ocp} activated at time $k$ is the minimization of the cost function 
\begin{multline}
    \VN(\xt_{\cdot|k},u_{\cdot|k},\mathcal{T}({\theta(k),\dec(k)}))  =\\  \sum_{\el \in \mathcal{N}^{\theta(k),\dec(k)}_{0:N-1}}  p_\el l(\xt_{\el|k},u_{\el|k}) + \sum_{\el\in\mathcal{N}^{\theta(k),\dec(k)}_N} p_\el \Vf(\xt_{\el|k},\mm_\el), \label{eq:cost_function}
\end{multline}
over the predicted state and input values $\xt_{\ell|k}, u_{\ell|k}$, where $l(\xt,u)=\xt^\top Q \xt + u^\top R u$ defines the stage cost with $Q, R \succ 0$ and $\Vf(\xt,\mm) = \xt^\top P_\mm \xt$ the terminal cost function with positive definite matrices $\SOL_\mm \in \mathbb{R}^{\nstatest\times\nstatest}$ for all $\mm \in \Theta$. 
The following standard assumption provides the terminal ingredients used to prove recursive feasibility and stability of the closed-loop system.
An approach to compute suitable terminal ingredients can be found, e.g., in \cite{patrinos2014mpc}.

\begin{assumption}
    There exists a mode-dependent terminal set $\Xf = \{\Xfi\}_{i\in\Theta} \subseteq \mathcal{X}$, terminal cost function $\Vf(\xt,i) = \xt^\top P_i \xt$ with positive definite matrices $\SOL_i \in \mathbb{R}^{\nstatest\times\nstatest}$, and static feedback controllers $K_i \in \mathbb{R}^{\nstatest\times\ninputs}$, such that $K_i\xt\in\mathcal{U}$, $(A_i+B_i K_i)\,\xt \in \Xf_j$ for all 
   $j\in\mathcal{S}_i = \{j \in \Theta \mid \bar{p}_{ij}>0\}$, and it holds true $\forall\xt\in\Xfi$:
    \begin{align*}
        l(\xt,K_i \xt) + \displaystyle{\sum_{j\in\Theta}} \, \bar{p}_{ij} \, \Vf \big((A_i+B_iK_i) \, \xt,j \big)
        \leq \Vf(\xt,i).
    \end{align*}
\label{assumption:terminal_set}
\end{assumption}
In the following, $\Phi_i = A_i+B_i K_i$ is used for convenience.
We define the \gls{ocp} for the tree $\stree{\theta(k),\dec(k)}$
as
\begin{subequations}\label{eq:ocp}
    \begin{alignat}{3}
     \VN^\star(\xt(k), \theta(k), d(k) &) \label{eq:ocp:cost}\\
     = \min_{z,u}\;\;
      & \VN(\xt_{\cdot|k},u_{\cdot|k},\mathcal{T}(\theta(k),\dec(k))) \nonumber
      &
      & \\
      \text{s.t.}\;\;\,
      & \xt_{0|k} = \xt(k), \label{eq:ocp_initial}\\ 
      & \xt_{\mathtt{succ}(\ell,i)|k} = A_{\theta_\el} \,\xt_{\el|k} + B_{\theta_\el} \,u_{\el|k}, \label{eq:ocp_evo}\\ 
      & H\, u_{\el|k} \leq h, \label{eq:ocp_constu}
      \\ 
      & G\, \xt_{\el|k} \leq g, \label{eq:ocp_constz}\\ 
      & \xt_{r|k} \in \mathcal{X}^\mathrm{f}_{\theta_r}, \label{eq:ocp_terminal}
    \end{alignat}
\end{subequations}
for all $i \in \mathcal{S}_{\theta_\el}$, $\el \in \mathcal{N}^{\theta(k),\dec(k)}_{0:N-1}$ and $r \in \mathcal{N}^{\theta(k),\dec(k)}_N$.

Here, \eqref{eq:ocp:cost} specifies the optimal value function, \eqref{eq:ocp_initial} encodes the initial conditions, \eqref{eq:ocp_evo} the system evolution. \eqref{eq:ocp_constu} deals with the actuator constraints and \eqref{eq:ocp_constz} with the state constraints. Finally \eqref{eq:ocp_terminal} encodes the terminal set constraint.

\subsection{Optimal hybrid overrun decision}
\label{sec:algo_hybrid_strategy}

\begin{algorithm}[t]
\small
\caption{Optimal hybrid overrun strategy search}
\label{alg:hybrid_strategy}
\begin{algorithmic}[1]
\Function{GET\_MINIMUM\_COST\_PATH}{$\xt^\star$, $u^\star$} \label{alg1:function_get_min_cost_path}
    \State shift $\xt^\star$, $u^\star$ from last solved OCP to new full tree \label{alg1:shift_x_u}
    \State min\_cost, min\_path = DFS(0) \Comment{\emph{call DFS for root node}} \label{alg1:dfs_0}
    \State \textbf{return} decision\_seq(min\_path) \Comment{\emph{output new decision $\dec^+$}}\label{alg1:return_min_path}
\EndFunction
\Function{DFS}{node $\el$} \label{alg1:function_dfs}
    \If{($\xt_\el$ is empty) or ($u_\el$ is empty)} \label{alg1:if_x_empty}
        \State $\xt_\el,u_\el \leftarrow$ propagate\_dynamics($\theta_\el$) \Comment{\emph{run $\eqref{eq:propagate}$}}
        \label{alg1:dynamics}
    \EndIf \label{alg1:end_if_x_empty}
    \State cost = cost\_func($\xt_\el,u_\el$) \Comment{\emph{evaluate cost of node $\el$}}\label{alg1:init_cost}
    \State path = [$\theta_\el$] \Comment{\emph{initialize list of path modes}}\label{alg1:init_path}
    \State cost += cost\_subtree($\mathtt{succ}(\el,\textsc{h})$) \Comment{\emph{cost of subtree after \textsc{h}}}\label{alg1:if_succ_h}
    \State child\_kc = [$\mathtt{succ}(\el,i)$ for $i$ in $\{\star\textsc{K},\star\textsc{C}\}$] \Comment{\emph{overrun nodes}} \label{alg1:get_child_ovr}
    \If{child\_kc is not empty} \label{alg1:if_child_ovr_not_empty}
        \For{j, child in enumerate(child\_kc)} 
            \State costs[j], paths[j] = DFS(child) \Comment{\emph{call DFS for child}}\label{alg1:dfs_child_ovr}
        \EndFor 
        \State cost += min(costs) \Comment{\emph{update with best child path cost}}\label{alg1:cost_plus_min_costs}
        \State path $\leftarrow$ path.extend(paths[argmin(costs)]) \Comment{\emph{app. path}}\label{alg1:path_plus_min_path}
    \EndIf \label{alg1:endif_child_ovr}
    \State \textbf{return} cost, path \label{alg1:return_node_cost_path}
\EndFunction
\end{algorithmic}
\end{algorithm}
After a control job successfully completes solving \eqref{eq:ocp} over $\stree{\theta,\dec}$, a new decision sequence $\dec^+$ of our hybrid overrun strategy must be computed. 
To streamline the presentation, we first introduce Algorithm \ref{alg:hybrid_strategy} for efficiently computing the optimal decision strategy, while the formal guarantees are provided in Section \ref{sec:guarantees} with Lemma \ref{lem:min_cost_decision}.
Algorithm \ref{alg:hybrid_strategy} uses as input the optimization variables $\xt^\star$, $u^\star$ from the solved \gls{ocp}, and outputs a new sequence $\dec^+$ that will be valid until the next completed job, when the strategy is updated again.

It is useful here to define the combined tree $\hat{\mathcal{T}}$ that is formed by merging all pruned trees with root mode $\theta_0 = \textsc{h}$, i.e. $\hat{\mathcal{T}} = \bigcup_{d\in\decset} \stree{d,\textsc{h}}$. The combined tree spans over all possible hybrid decision sequences, while all nodes with the same mode path from the root node are merged into one.

Let $k$ be the step immediately after the execution of the Algorithm.
First, the combined scenario tree $\hat{\mathcal{T}}$ starting at $\textsc{h}$ is assigned the candidate values $\hat{\xt}$, $\hat{u}$ as the properly shifted optimized variables $\xt^\star$, $u^\star$ from the completed \gls{ocp}, using the mapping function $g$ with the path $[\theta_0,\dots,\theta(k-1),\textsc{h}]$ followed in the previous tree (line \ref{alg1:shift_x_u}).
Subsequently, the routine for the depth first search (DFS) of the optimal hybrid strategy is entered, starting at the root node with index $0$ (line \ref{alg1:dfs_0}).
The function DFS($\el$) first assigns to node $\el$ a propagated state and input values if not available from $\xt^\star$, $u^\star$ (lines \ref{alg1:if_x_empty}-\ref{alg1:end_if_x_empty}). Such computation is performed as follows:
\begin{align}
\label{eq:propagate}
    \hat{u}_{\ell|k} = K_{\theta_{\ell}}^{\vphantom{\star}} \hat{\xt}_{\ell|k}, \qquad
    \hat{\xt}_{\ell|k} =  \Phi_{\theta_{\mathtt{pred}(\ell)}}\, \hat{\xt}_{\mathtt{pred}(\ell)|k},
\end{align}
where the function $\mathtt{pred}(\el)$ outputs the predecessor node of $\el$.
Line \ref{alg1:init_cost} ensures that the cost variable is initialized with the node's cost value, i.e., either the stage cost $l(\xt_\el,u_\el)$ or the terminal cost $\Vf(\xt_\el,\theta_\el)$, multiplied by its probability $p_\el$.
The path variable is initialized with the mode $\theta_\el$ (line \ref{alg1:init_path}).
Since the optimal path search algorithm is only considering overrun modes,
subtrees starting with $\textsc{h}$ are not traversed in the optimal path search (aside from the initial node). 
Instead, their costs are added to the predecessor node's cost (line \ref{alg1:if_succ_h}).
Then, the function iterates by calling DFS for all successor nodes with modes in $\{\star \textsc{k},\star \textsc{c}\}$ (if any, line \ref{alg1:dfs_child_ovr}). The cost of the optimal path is added to the local cost variable, and the path with minimum cost is appended to the local path variable (lines \ref{alg1:cost_plus_min_costs}-\ref{alg1:path_plus_min_path}), finally mapped to the corresponding decision sequence $\dec^+$ once all subroutines end.

\subsection{Feasibility and mean-square stability}
\label{sec:guarantees}
In the following section we present our theoretical results on recursive feasibility (Lemma \ref{lem:recursive_feasible}), optimality of the decision sequence (Lemma \ref{lem:min_cost_decision}) and mean-square stability of the closed-loop system (Theorem \ref{theorem:stability}).

\begin{lemma}
\label{lem:recursive_feasible}
    Let Assumption~\ref{assumption:terminal_set} hold and consider the MJLS~\eqref{eq:mjls} under control law $u(k) = u^\star_{0|k}$ resulting from \eqref{eq:ocp}. If the optimization problem \eqref{eq:ocp} is feasible at time $k=0$ and the hybrid strategy $d$ is not changed until termination of a control job, then the optimization problem~\eqref{eq:ocp} is recursively feasible for all $k>0$.
\end{lemma}
\begin{proof}
    The proof uses similar argument as in \cite{patrinos2014mpc} and \cite{gallant2025softMPC}, adapted for the hybrid strategy decision case.
    
    Assume that at time $k$ the optimal state and input trajectories $\xt^\star_{\cdot|k},u^\star_{\cdot|k}$ resulting from MPC optimization problem \eqref{eq:ocp} exist. 
    At time $k+1$, we construct a feasible candidate input
    \begin{align*}
     \hat{u}_{\ell|k+1} &= u^\star_{\Gamma(\ell)|k} \quad && \forall \ell \in \mathcal{N}^{\theta(k+1),d}_{0:N-2}, \\
     \hat{u}_{\ell|k+1} &= K_{\theta_{\ell}}^{\vphantom{\star}}\xt^\star_{\Gamma(\ell)|k}\quad && \forall \ell \in \mathcal{N}^{\theta(k+1),d}_{N-1},
    \end{align*}
    where the mapping function $\Gamma$ depends on the previous decision $d$ and the mode sequence $[\theta(k), \theta(k+1)]$. 
        
    Due to feasibility at time $k$, the shifted input sequence satisfies the constraints \eqref{eq:ocp_constu} for all stages $\ell \in \mathcal{N}^{\theta(k+1),d}_{0:N-2}$. From Assumption \ref{assumption:terminal_set} the terminal input fulfills $\hat{u}_{\ell|k+1} \in \mathcal{U}$, verifying \eqref{eq:ocp_constu} at stage $\ell \in \mathcal{N}^{\theta(k+1),d}_{N-1}$.
    
    Applying $\hat{u}_{\ell|k+1}$ to the dynamics \eqref{eq:ocp_evo} results in the shifted state sequence
    \begin{alignat*}{3}
     \hat{\xt}_{\ell|k+1} &= \xt^\star_{\Gamma(\ell)|k} \quad \quad && \forall \ell \in \mathcal{N}^{\theta(k+1),d}_{0:N-1}, \\ 
    \hat{\xt}_{\ell|k+1} &=  \Phi_{\theta_{\mathtt{pred}(\ell)}}^{\vphantom{\star}} \, \hat{\xt}_{\mathtt{pred}(\ell)|k+1}  \qquad && \forall \ell \in \mathcal{N}^{\theta(k+1),d}_{N}.
    \end{alignat*}

    In view of feasibility at time $k$, the state constraints \eqref{eq:ocp_constz} are verified for all stages $\ell \in \mathcal{N}^{\theta(k+1),d}_{0:N-1}$.

    From invariance of the terminal set (Assumption \ref{assumption:terminal_set}), it follows that $\hat{\xt}_{\ell|k+1} \in \mathcal{X}^\mathrm{f}_{\theta_\ell}$, which verifies \eqref{eq:ocp_constz} and \eqref{eq:ocp_terminal} at stage $\ell \in \mathcal{N}^{\theta(k+1),d}_{N}$. This concludes the proof.
\end{proof}

\begin{lemma}\label{lem:min_cost_decision}
    Assume Lemma \ref{lem:recursive_feasible} holds. Let $\VNhat(\hat{\xt},\hat{u},\stree{\theta,d})$ denote \eqref{eq:cost_function} evaluated with candidate sequences $\hat{\xt},\hat{u}$, i.e., the suboptimal value function of the MPC optimization problem \eqref{eq:ocp}.
    Then there exists a predictive decision sequence $d^{+}$ and corresponding tree $\stree{\textsc{h},d^+}$ that satisfies
    \vspace{-0.4em}
    \begin{align*}
        \VNhat(\hat{\xt},\hat{u}, \stree{\textsc{h}, \dec^+}) \leq \VNhat(\hat{\xt},\hat{u}, \stree{\textsc{h}, \dec}) \quad \forall \dec \in \decset.
    \end{align*}
\end{lemma}
\smallskip
\begin{proof}
    By Lemma \ref{lem:recursive_feasible} there exist feasible candidate sequences $\hat{\xt}, \hat{u}$ at all time steps $k > 0$.
    Each tree $\mathcal{T}_{\dec} \in \mathcal{F}(\textsc{h})$ is populated with its feasible candidate sequences $\hat{\xt}$, $\hat{u}$ resulting in the tree-dependent cost $\VNhat(\hat{\xt},\hat{u}, \mathcal{T}_\dec)$.
    By construction
    \vspace{-0.4em}
    \begin{align*}
        \VNhat(\hat{\xt}, \hat{u}, \mathcal{T}(\textsc{h}, \dec^+)) &= \min_{\mathcal{T}_d \in \mathcal{F}(\textsc{h})} \VNhat(\hat{\xt},\hat{u}, \mathcal{T}_{\dec}) \\
        &\leq \VNhat(\hat{\xt}, \hat{u}, \mathcal{T}(\textsc{h}, \dec)) \quad \forall \dec \in \decset,
    \end{align*}
    where $d^+ \in \mathcal{D}$ is the optimal predictive decision sequence associated with the minimizer $\mathcal{T}_{d^+}$.
\end{proof}

The proposed Algorithm \ref{alg:hybrid_strategy} is an efficient way to find the minimal cost decision sequence. Instead of evaluating all $2^{n_\mathrm{M}}+n_\mathrm{M}\cdot 2^{n_\mathrm{M}-1}$ pruned trees separately, only one combined tree is searched with an algorithm complexity of $\mathcal{O}(\numNodes)$ and $\numNodes = \mathcal{O}(3^N)$.
The following theorem, our main result, establishes mean-square stability of the closed-loop system under the hybrid strategy from Algorithm \ref{alg:hybrid_strategy}.
\begin{theorem}\label{theorem:stability}
    Let Assumption \ref{assumption:terminal_set} hold and consider the MJLS \eqref{eq:mjls} under control law $u(k) = u^\star_{0|k}$ resulting from \eqref{eq:ocp}. 
    If the \gls{ocp} \eqref{eq:ocp} is feasible for $\xt(0)$ and the hybrid strategy decision is computed by Algorithm \ref{alg:hybrid_strategy}, then the origin of the closed-loop system is mean-square stable.
\end{theorem}
\begin{proof}
    To prove the claim, we first show similar to  \cite{gallant2025softMPC} stochastic contractivity of the optimal value function: 
    \begin{multline}
        \mathbb{E}[ \VNstar(\xt^+,\theta^+,\dec^+)
        - \VNstar(\xt,\theta,d) \mid \theta ] \\
        = {\displaystyle\sum_{j\in\Theta}} p_{\theta j}(d) \VNstar(A_{\theta}x +B_\theta u_0^\star,j,\dec^+) - \VNstar(\xt,\theta,\dec), \label{eq:theorem1:value_decrease_sum}
    \end{multline}
    We consider the suboptimal cost $\VNhat$, i.e., \eqref{eq:cost_function} evaluated with the candidate values $\hat{\xt}$ and $\hat{u}$ from Lemma \ref{lem:recursive_feasible} and account for the updated decision sequence in the event of a deadline hit by factoring out mode $\textsc{h}$ from the summation, resulting in
    \begin{alignat}{3}
        \eqref{eq:theorem1:value_decrease_sum} \leq \; &p_{\theta \textsc{h}} \VNhat(\hat{\xt},\hat{u},\mathcal{T}(\textsc{h},d^+)) \nonumber \\
        & + {\textstyle\sum_{j\in\Theta\backslash \{\textsc{h}\}}} p_{\theta j}(d) \VNhat(\hat{\xt},\hat{u},\mathcal{T}(j,d)) \nonumber\\
        & - \Bigl[ l(\xt,u_0^\star) + 
        \Bigl(
             {\textstyle\sum_{\el \in \mathcal{N}^{\theta,\dec}_{1:N-1}}}  p_\el l({\xt}_\ell^\star,{u}_\ell^\star) \nonumber \\
             &\hspace{2.3cm} + {\textstyle\sum_{\el\in\mathcal{N}^{\theta,\dec}_N}} p_\el \Vf({\xt}_\ell^\star,\mm_\el)
        \Bigr) \Bigr]. \label{eq:theorem1:mss_proof_sep}
    \end{alignat}
    
    With Lemma \ref{lem:min_cost_decision} we have
    \begin{alignat*}{3}
        \VNhat(\hat{\xt},\hat{u},\mathcal{T}(\textsc{h},d^+)) 
        \leq \VNhat(\hat{\xt},\hat{u},\mathcal{T}(\textsc{h},d))
    \end{alignat*}
    and thus bound the cost decrease by
    \begin{subequations}
    \begin{alignat*}{3}
        \eqref{eq:theorem1:mss_proof_sep} \,&\leq -l(\xt,u_0^\star) + {\textstyle\sum_{\ell \in \mathcal{N}^{\theta,\dec}_{N}}} p_\ell \bigl[ -\Vf(\xt_\ell^\star,\theta_\ell)
          \\[-0.1cm]
         & \quad + l(\xt_\ell^\star,K_{\theta_\ell} \xt_\ell^\star) + {\textstyle\sum_{m\in\mathcal{S}_{\theta_\ell}}} \bar{p}_{\theta_\ell \theta_m} \Vf(\Phi_{\theta_\ell} \xt_\ell^\star, \theta_m) \bigr] \\
         &\overset{\text{Ass. } \ref{assumption:terminal_set}}{\leq} - l(\xt,u_0^\star) \leq - \Vert \xt \Vert_Q^2.
    \end{alignat*}
    \end{subequations}
    By Assumption~\ref{assumption:terminal_set} and Lemma~\ref{lem:recursive_feasible} the MPC optimization problem \eqref{eq:ocp} is recursively feasible. Furthermore, we can establish the asymptotic average cost bound
    \begin{align*}
       & \lim_{t\rightarrow \infty}  \nicefrac{1}{t}\, \mathbb{E}[ \VNstar(\xt(t),\theta(t),\dec(t) - \VNstar(\xt(0),\theta(0),d(0))] \\
         \leq&  \lim_{t \rightarrow \infty} \nicefrac{1}{t} {\textstyle \sum_{k = 0}^t} -\mathbb{E}[\Vert \xt(t) \Vert_Q^2],
    \end{align*}
   which, after rearranging the terms shows that
   \begin{align*}
    \lim_{t \rightarrow \infty} \frac{1}{t} \sum_{k = 0}^t \mathbb{E}[\Vert \xt(t) \Vert_Q^2] \leq \lim_{t \rightarrow \infty} \frac{1}{t} \VNstar(\xt(0),\theta(0),d(0)) = 0.   
   \end{align*}
    Since $Q \succ 0$, this implies that $\mathbb{E}(\Vert \xt(t) \Vert^2) \rightarrow 0$ as $t \rightarrow \infty$, which proves mean-square stability of the origin~\cite{patrinos2014mpc} and concludes the proof.
\end{proof}

\section{Evaluation}

We evaluate the proposed approach on two different linear systems, under different configurations.
For both case studies we sample the response times of each job from a lognormal distribution with $\mu=0$ and $\sigma = 0.6$, multiplied by sampling time $t_\mathrm{s}$ and clipped to $(0, (n_\mathrm{M}+1) t_\mathrm{s})$. 
Following our model, the modes of the switched system are determined based on the active overrun strategy.
The maximum value $n_\mathrm{M}$ is enforced in our scheduling simulations, by shortening the response time of the job that would violate it.

\subsection{Double integrator system}
\label{sec:evaluation:integrator}
First we evaluate our approach on a double integrator system with $t_\mathrm{s}=1\,\mathrm{s}$, $n_\mathrm{M}=2$ and actuator strategy \textsc{zero}, with nominal dynamics
$\scriptsize A = \colvec{1 & 1 \\ 0 & 1}$ and $\scriptsize B = \colvec{0.5 \\ 1}.
$
We consider an additive disturbance resulting in the dynamics
\begin{align*}
    \xt(k+1) = A_{\theta(k)} \,\xt(k) + B_{\theta(k)}\, u(k) + W w(k),
\end{align*}
where $W = \colvec{0 & 1 & 0 & 0}^\top$ and $w(k) = 0.5 \cdot (-1)^{k\!\!\pmod{2}}$ for $k\in[3,11]$ and $w(k) = 0$ otherwise.
We use the parameters $N=6$, $Q=10^{-2}\cdot I^{\nstatest \times \nstatest}$ and $R=5\cdot10^{-2}$ and evaluate the control performance using the cost function
    $J(k) = \xt(k)^\top Q \xt(k) + u(k)^\top R u(k)$, summed over a given interval.
We compare our adaptive approach against an MPC design as in \eqref{eq:ocp} but using fixed strategy ($\textsc{csn}$ or $\textsc{kll}$). 
Each design is evaluated across $100$ Monte Carlo simulation runs.

\begin{figure}
    \centering
    \includegraphics{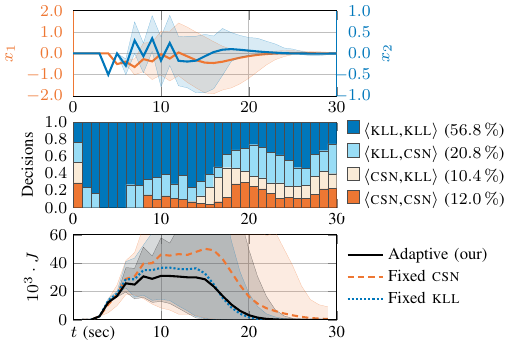}
    \vspace{-0.3cm}
    \caption{Mean and $0.01$ to $0.99$ quantiles of: (top) state trajectories under our hybrid adaptive strategy, (middle) moving window cost (window size $n_\mathrm{M}+1$) for our, fixed $\textsc{csn}$ and $\textsc{kll}$. Percentage of selected decision sequences over time (bottom).}
    \vspace{-3mm}
    \label{fig:integrator}
\end{figure}

Figure \ref{fig:integrator} shows the state trajectories of the double integrator controlled by our hybrid overrun \gls{mpc}; the ratio of selected hybrid decisions at each time step; and the moving window costs obtained with our design, against the fixed ones.
The legend of the mid plot shows the total percentage of each decision across the entire simulation (excluding the initial randomized decision $\dec(0)$).

The results show that strategies are successfully adapted at runtime.
Overall, the mean total cost with fixed strategy $\textsc{csn}$ and $\textsc{kll}$ is $ 0.219$ and $0.144$, respectively, while our hybrid adaptive strategy has total cost of $0.129$, which is a $41.1\,\%$ / $10.4\,\%$ improvement compared to $\textsc{csn}$ / $\textsc{kll}$. 
Similar results (with detailed graphs omitted for space reasons) have been obtained for the same setup with $n_\mathrm{M}=1$ ($44.1\,\%$ / $10.0\,\%$ improvement) and with $n_\mathrm{M}=3$ ($36.3\,\%$ / $9.7\,\%$).

\subsection{Aircraft pitch system}

As a second example, we consider an aircraft pitch control system similar to \cite{hobbs2022safety}, with $t_\mathrm{s}=0.2\,\mathrm{s}$, yielding the nominal dynamics
\begin{align*}\footnotesize
    A = \colvec{\hphantom{-}0.9247 & 10.4772 & 0 \\
                -0.0026 & \hphantom{1}0.9038 & 0 \\
                -0.0150 & 10.8142 & 1}, \quad
    B = \colvec{0.0666 \\ 0.0038 \\ 0.0221}.
\end{align*}
The system states are $x^\top = \colvec{\alpha & q & \phi}$, with $\alpha$ being the angle of attack, $q$ the pitch rate and $\phi$ the pitch angle. The system input is the elevator deflection angle $\delta$. We apply the limit $|\delta|\leq 0.5 \, \mathrm{rad}$ and safety constraint $|\phi| \leq 1\,\mathrm{rad}$, and use
$Q = \mathrm{diag}\,(10^{-2}, 10^{-2}, 0.2, 10^{-2}, 10^{-2})$ and $R=10^{-2}$.
We add an impulse disturbance of $w(k)=\colvec{0 & 0& 0.3 & 0& 0}^\top$ at $k=6$ and $k=20$. 
During $k \in [20,40]$ in addition to the impulse disturbance we simulate a response time disturbance, where the individual control jobs' response times are assigned a constant value randomly sampled from the lognormal distribution, but multiplied by $2$, over all remaining job instances of the simulation.
Again, we compare against the MPC with fixed strategies (\textsc{csn} or \textsc{kll}) and evaluate across $100$ Monte Carlo simulation runs for each design.

\begin{figure}
    \centering
    \includegraphics{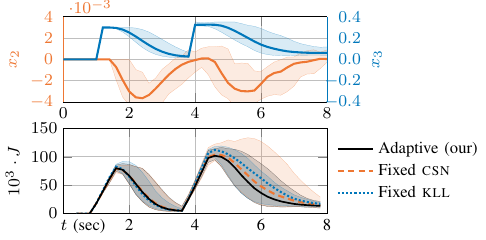}
    \caption{Mean and $0.01$ to $0.99$ quantiles of: (top) state trajectories under our hybrid adaptive strategy; (bottom) moving window cost (window size $n_\mathrm{M}+1$) for our, fixed $\textsc{csn}$ and $\textsc{kll}$.}
    \vspace{-2mm}
    \label{fig:aircraft_trajectories}
\end{figure}

\begin{figure}[t]
    \centering
    \includegraphics{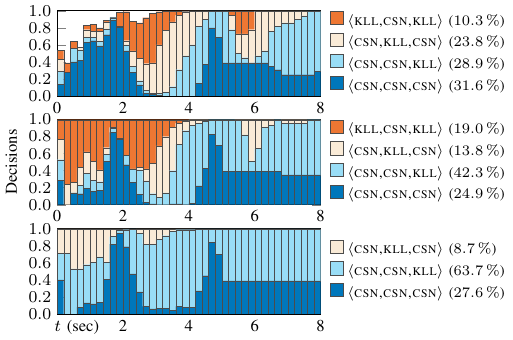}
    \vspace{-0.3cm}
    \caption{Percentage of selected decision sequences for: (top) the full hybrid adaptive approach showing only the most frequent four out of eight alternatives, (middle) pruned approach (a), (bottom) pruned approach (b).}
    \vspace{-4mm}
    \label{fig:aircraft_pruned}
\end{figure}

Figure \ref{fig:aircraft_trajectories} shows $q$ and $\phi$ under our proposed control approach when $n_\mathrm{M}=3$, and the moving window costs evaluated at each time step for our hybrid adaptive and the two fixed designs. 
The total cost's mean $\mu$ and variance $\sigma^2$ are reported in Table \ref{tab:costs_time} for different $n_\mathrm{M}$.
Overall, safety constraints are always guaranteed, and we observe better control performance with our approach for $n_\mathrm{M}>1$, compared to the fixed \textsc{csn}/\textsc{kll} controllers.

\subsection{Pruned hybrid overrun approach}
\label{sec:pruned_hybrid}

Table \ref{tab:costs_time} reports also the average execution times $\Delta t$ (recorded on an i7-1270P processor using MOSEK via YALMIP \cite{mosek,lofberg2004yalmip}) of solving \eqref{eq:ocp} for different $n_\mathrm{M}$, which represents the most computationally intensive component of our design ($>90\,\%$ of total runtime).

The computation times for our hybrid approach are mainly affected by the increasing size of $|\decset|$ with $n_\mathrm{M}$.
To counteract this, we propose a proper \emph{offline pruning} of the possible decision sequences, obtaining $\bar{\decset}\subset\decset$, based on the observation that some sequences are considerably more frequently selected than others, particularly for high $n_\mathrm{M}$. 
An example is shown in the top graph of Figure \ref{fig:aircraft_pruned}, where the four most frequent (out of $8$ total) sequences sum up to $94.6 \,\%$ of total decisions.
Such pruning does not significantly alter the feasibility and mean-square stability arguments, given that a terminal controller exists that satisfies Assumption \ref{assumption:terminal_set} with $\bar{p}_{ij}$ applied for transitions with open decisions, and $p_{ij}$ applied for transitions where the opposing decision has been made impossible by the pruning.

The costs and computation times for different pruned approaches are reported on the right side of Table \ref{tab:costs_time}.
Figure \ref{fig:aircraft_pruned} shows the corresponding frequency of selected hybrid strategies for the full and two pruned approaches (a) and (b) for $n_\mathrm{M}=3$. 
Here, for pruned (a), $\bar{\decset}$ includes only the four most frequent sequences of ${\decset}$, while for pruned (b) only the three most frequent ones ($84.3 \,\%$ of decisions).
For $n_\mathrm{M}=2$ (only shown in the Table) $\bar{\decset}$ includes the two most frequent sequences ($91.9\,\%$).
For $n_\mathrm{M}=4$ we limit $\bar{\decset}$ in (a) to sequences with first decision \textsc{csn} ($85.7\,\%$) and (b) to the two most frequent sequences ($67.9\,\%$).
The results show that careful pruning can significantly decrease the computation time with minimal loss in performance. Also, multiple choices are possible based on the acceptable threshold of such trade-off. 
In particular the pruned (b) case for $n_\mathrm{M}=4$ highlights the potential of such pruning approach, with a $71\,\%$ reduction of computation time with respect to the adaptive full design, and a loss of performance of only $2.7 \,\%$.

\begin{table}[t]
\renewcommand{\arraystretch}{0.95}
    \caption{Total costs $\times 10^{3}$ and OCP solving times in $\mathrm{ms}$ using fixed strategies \textsc{csn} and \textsc{kll}, our adaptive approach (Adapt.) and pruned approaches (Prd. (a), Prd. (b)).}
    \vspace{-2ex}
    \small
    \centering
        \begin{tabular}{p{0.008\linewidth}p{0.02\linewidth}p{0.10\linewidth}p{0.11\linewidth}p{0.14\linewidth}!{\color{black!30}\vrule width 0.2pt}p{0.12\linewidth}p{0.125\linewidth}}
        \toprule
        $\!n_\mathrm{M}$&&{ \textsc{csn}} & {\textsc{kll}} & \textbf{Adapt.} & \textbf{Prd. (a)} & \textbf{Prd. (b)} \\
        \midrule
        $1$ & { $\mu$}  & 335.1 & \textbf{311.7} & 332.7 & & \\
        & { $\sigma^2$}& 0.25 & 0.18 & 0.28 &&  \\
        $2$ & { $\mu$} & 382.9 & 370.5 & \textbf{349.4} & 356.4 &  \\
        & { $\sigma^2$}& 2.73 & 0.71 & 0.83 & 0.91 & \\
        $3$ & { $\mu$} & 413.4 & 448.4 & \textbf{381.1} & 382.9 & 390.5 \\
        & { $\sigma^2$} & 8.3 & 2.6 & 3.3 & 4.0 & 3.8 \\
        $4$ & { $\mu$} & 429.6 & 539.2 & \textbf{398.7} & 402.2 & 409.3 \\ 
        & { $\sigma^2$} & 15.3 & 7.0 & 7.5 & 7.7 & 7.8 \\
        \midrule
        1 &{ $\Delta t$}& 12.6 & 13.2 & 30.3 &  & \\
        2 &{ $\Delta t$}& 19.8 & 22.2 & 78.9 & 25.3 & \\
        3 &{ $\Delta t$}& 22.9 & 23.4 & 109.4 & 86.8 & 45.8 \\        
        4 &{ $\Delta t$}& 26.7 & 33.4 & 142.5 & 105.0 & 41.3 \\
        \bottomrule
        \end{tabular}
    \label{tab:costs_time}\\[2ex]
    \vspace{-5mm}
\end{table}
\section{Conclusion}

This paper introduces the first adaptive hybrid overrun approach for real-time control systems, which iteratively determines at runtime the best predictive combination of overrun strategies, between \emph{kill} and \emph{continue-and-skip-next}, for the next iterations.
We combine a scenario-tree MPC and a DFS-based algorithm to make optimal decisions in terms of both overrun strategy decision sequence and associated control input.
The proposed control policy and adaptive overrun strategy guarantee stability and feasibility for a system subject to state and input constraints.
We show the control performance benefits of our approach on multiple relevant use-cases, up to $44\,\%$ for the double integrator and $26\,\%$ for the aircraft use case. Finally, a pruned design is proposed to reduce the computational runtime (by up to $71\,\%$) with minimal performance loss ($2.7 \,\%$).

\bibliographystyle{IEEEtran}
\bibliography{biblio}

\end{document}